\documentclass[12pt]{article}
\usepackage[letterpaper, margin=1in]{geometry}
\usepackage{amsmath, amsthm, amssymb,amsfonts, graphicx, subfigure}
\usepackage[colorlinks,citecolor=blue,linkcolor=blue, urlcolor=blue, anchorcolor=blue]{hyperref}
\usepackage{makecell}
\usepackage{dcolumn}
\usepackage{mathrsfs}
\usepackage{float}
\usepackage{multirow}
\usepackage{braket}
\usepackage{booktabs}
\usepackage{color}
\usepackage{tcolorbox}
\usepackage[ruled]{algorithm2e}
\usepackage{tikz}
\usepackage{quantikz}
\usepackage{authblk}
\usepackage{enumerate}
\newtheorem{lemma}{Lemma}
\newtheorem{theorem}{Theorem}
\newtheorem{corollary}{Corollary}
\newtheorem{remark}{Remark}
\newtheorem{definition}{Definition}
\newtheorem{fact}{Fact}
\newtheorem{example}{Example}

\begin{document}
	\title{Unifying quantum spatial search,  state transfer and uniform sampling  on graphs: simple and exact }
	
	\author{Qingwen Wang}
       \author{Ying Jiang}
       \author{Lvzhou Li\thanks{Email: lilvzh@mail.sysu.edu.cn}}
	\affil{School of Computer Science and Engineering, Sun Yat-sen University, Guangzhou 510006, China}

	\maketitle
	\begin{abstract}
		This article presents a novel and succinct algorithmic framework via alternating quantum walks, unifying  quantum spatial search, state transfer  and uniform sampling on a large class of graphs. Using the  framework, we can achieve exact uniform sampling over all vertices and perfect state transfer between any two vertices, provided that eigenvalues of Laplacian matrix of the graph are all integers. Furthermore, if the graph is vertex-transitive as well, then we can  achieve deterministic quantum spatial search that finds a marked vertex  with certainty. In contrast, existing quantum search algorithms generally has a certain probability of failure. Even if the graph is not vertex-transitive, such as the complete bipartite graph, we can still adjust the algorithmic framework to obtain deterministic spatial search, which thus shows the flexibility of it. Besides unifying and improving  plenty of previous results, our work provides new results on more graphs. The approach  is easy to use since it has a succinct formalism that depends only on the  depth  of the Laplacian eigenvalue set of the graph, and may shed light on the solution of more problems related to graphs.
		
		
	\end{abstract}
	\section{Introduction}
	Quantum walks have developed
	into a crucial and useful primitive for quantum algorithm design. Since   Aharonov, Davidovich and  Zagury ~\cite{PhysRevA.48.1687} first introduced the term ``quantum walks'' in 1993, quantum walks have become one of the key components in quantum computation ~\cite{Kempe_overview,venegas-andraca_quantum_2012,Kadian2021}. Quantum walks can be divided into two main types: discrete-time quantum walks (DTQWs) and continuous-time quantum walks (CTQWs). Whereas  CTQWs evolve a Hamiltonian $H$ (related to the graph under consideration) for any time $t$, i.e. simulating $e^{iHt}$, DTQWs  evolve the system for discrete time steps, i.e. applying $U_\mathrm{walk}^t$ to the initial state for some integer $t$ and unitary operator $U_\mathrm{walk}$.   DTQWs can be further categorized into many different frameworks. The earliest and simplest model of DTQWs is the coined quantum walk~\cite{10.1145/380752.380757, AmbainisKV01}. Subsequently, Szegedy proposed a quantum walk framework ~\cite{1366222} from the perspective of Markov chains.
	In this direction, a series of variant  frameworks for spatial search have been developed:  the MNRS framework~\cite{MNRS}, the interpolated walks~\cite{interpolated}, the electric network framework~\cite{electric,electricfind}  and so on. 
	Quantum walk-based algorithms   have provided   polynomial and even exponential speedups over classical algorithms.
	Typical examples include quantum algorithms for the element distinctness problem \cite{Ambainis07}, matrix product verification \cite{BuhrmanS06}, triangle finding \cite{MagniezSS07}, group commutativity \cite{MagniezN07}, and the welded tree problem~\cite{10.1145/780542.780552,multi, li2024recovering}. It is worth mentioning that while it is habitually believed that  DTQW-based quantum algorithms  can provide only at most a quadratic speedup over classical algorithms, Li, Li and Luo \cite{li2024recovering} have proposed a succinct quantum algorithm based on the simplest coined quantum walks for the welded tree problem, which not only achieves exponential quantum speedups, but also returns the correct answer deterministically.

	\subsection{ Quantum spatial search}
	
	One of the most important classes of quantum walk-based algorithms is used to tackle the spatial search problem of finding unknown marked vertices on a graph. Over the past 20 years, the research on quantum spatial search algorithms can be roughly summarized into two lines. The first line is to search on  specific graphs, aiming at designing a quantum  algorithm with  time $O(\sqrt{N})$  for a given graph with $N$ vertices. It has been  investigated on plenty of different graphs such as   $d$-dimensional grids~\cite{coins}, hypercube graphs~\cite{RN2,hypercubegraph}, strongly regular graphs~\cite{Janmark2014}, complete bipartite graphs~\cite{RN10,bipartite,PhysRevA.106.052207}, balanced trees~\cite{RN3}, Johnson graphs \cite{RN4,johnson_discrete}, Hamming graphs \cite{HG}, Grassmann graphs~\cite{HG},   Erd\H{o}s-Renyi random
	graphs \cite{Chakraborty2016}, and so on. The  quantum spatial search algorithms for those graphs can be designed via  DTQWs or CTQWs. It is worth noting that discrete and continuous  walks have essential differences in the quantum setting, so designing a quantum algorithm in one model does not necessarily lead to a quantum algorithm in the other model. The relationship between  DTQWs and CTQWs can be found in the literature~\cite{Childs2010}. 
	
	In the above work,  the general method for analyzing the quantum algorithm is as follows: (i) reduce the state space to a low-dimensional invariant subspace, (ii) obtain the matrix representations of  the walk operators  in the invariant subspace, and (iii) derive the spectral decomposition of the matrices and use the eigenvalues and eigenvectors to analyze the success probability and  complexity of the  algorithm.  Each of the above steps is closely related to topological properties of the graph, so it usually needs to analyze the specific graphs case by case and it is generally difficult to generalize the results from one type of graphs to another. Naturally, one might ask: Is it possible to obtain a concise and universal quantum search method? 

	The second line of research on quantum spatial search  is not limited to specific graphs but is based on Markov chains to answer a more general question: Can quantum walks always provide a quadratic speedup over  classical walks for the spatial search problem? A breakthrough result on this problem was proved in 2020~\cite{10.1145/3357713.3384252}: For any graph, if there exists a classical search algorithm with time $O(T)$, then we can construct  a quantum search algorithm with time $\widetilde{O}(\sqrt{T})$. This result was proven based on DTQWs. Subsequently,  a similar result based on CTQWs  was proven in~\cite{PhysRevLett.129.160502}. These two results have provided a more comprehensive understanding on the application of quantum walks to spatial search problems. However, it should be noted that the research along the second line  cannot  replace the one along the first line. The reason  is  that  despite these elegent results~\cite{10.1145/3357713.3384252, PhysRevLett.129.160502}, when dealing with specific graphs, the optimal quantum search algorithm remains unknown, and we still need to fully utilize  topological properties of the graph to design algorithms.

	\subsection{Derandomization}
	
	Note that   in both the first and second lines of research mentioned above,  almost all of the  quantum search algorithms are not deterministic  (exact), that is, there is a certain   probability of failure.  This leads us to ask the question: Can these quantum search algorithms be improved to be exact  in principle,  without sacrificing quantum speedups?  In order to see the significance of this question, let us have a brief review  on the research of derandomization, which is the process of taking a randomized algorithm and turning it into a deterministic algorithm.  In the field of classical computing, randomness has long been regarded as an important resource to improve algorithm's efficiency. For instance, randomized algorithms can be significantly efficient in some choice of computational resource for a lot of basic computational problems, such as time for primality testing~\cite{miller76, Rabin80, SS77}, space for undirected s-t connectivity~\cite{AKLLR79} and circuit depth for perfect matching~\cite{KUW86}. 
	However, since the  polynomial-time deterministic algorithm for primality testing~\cite{AKS04} was proposed in 2004, the study of derandomization has been attracting continued attention in the field of theoretical computer science, see for instance, Refs.~\cite{ Russ06, Rein08, ST17, algebraic19, hard21, minimal23}. 
	There are also entire books on derandomization~\cite{pairwise06, pseudo12, quanti22}. However, the derandomization of quantum algorithms  lacks  theories and methods, despite some effort (e. g., \cite{QAA,arbi_phase,PhysRevA.64.022307,RN7,li2023derandomization, LI2024114551,li2024recovering}). 
	The significance of derandomizing quantum algorithms lies not only in theoretically proving that quantum algorithms can be improved to be exact in principle, but also in attracting interest from experimental scientists, e.g., \cite{li2023experimental,RN8,liu2015first}.

	\subsection{State transfer}	
	Quantum walks were also applied to the task of state transfer between two vertices of a graph. More specifically,  transfer from a vertex $u$ to another  vertex $v$ on a graph  is to  construct a  unitary operator $U$ such that $|\langle v|U|u\rangle|$ is as close to 1 as possible, where $U$ is the dynamics of a quantum walk on the graph. If  $|\langle v |U|u\rangle|=1$, then we say the state transfer is perfect. The form of $U$ depends on which quantum walk model used as described in~\cite{Portugal2018}. For example, when using CTQW as the walk model, $U=e^{-iHt}$  where $H$ is the adjacency matrix or Laplacian matrix of $G$.   Perfect state transfer (PST)  was originally introduced by Bose in~\cite{PhysRevLett.91.207901}. Since it has  wide utilization in quantum information processing~\cite{pst1,PhysRevA.71.032312,PhysRevA.78.052320,doi:10.1142/S0219749909006085,pst2,johnsonpst},  the problem of what graphs permit PST  has drawn much attention. 
	

	\subsection{ Uniform sampling}	
	The last problem considered in this work is quantum  sampling, which  refers to generating a quantum state corresponding to a probability distribution. More generally, quantum state preparation is a fundamental and crucial  issue in quantum computing~\cite{shende2005synthesis, plesch2011quantum, sun2023asymptotically, yuan2023optimal,luo2024circuit}.	The uniform superposition state is required in plenty of  quantum algorithms. For example,   in the   CTQW-based quantum spatial search algorithms  initiated by  Childs and Goldstone~\cite{RN2}, the uniform state $\ket{s}=\frac{1}{\sqrt{|V|}}\sum_{v\in V}\ket{v}$ over all vertices on a graph $G=(V,E)$ 
	is usually taken as the initial state  in the quantum algorithms,  without illustrating how to prepare it. There have been  some works \cite{interpolated,PhysRevA.102.022423,PhysRevA.107.022432} considering how to approximate the uniform state as close as possible, that is, generate a state $\ket{s'}$ such that $\big|\big|\ket{s'}-\ket{s}\big|\big|<\epsilon$, where the algorithm's complexity is  proportional to $\frac{1}{\epsilon}$. If $\epsilon=0$, then the process is called exact uniform sampling. Obviously, the approximate approach  does not apply to the exact case.

	In this work, we shall present a novel and succinct algorithmic framework unifying quantum spatial search, perfect state transfer, and exact uniform sampling on a large class of graphs. To this end, we need the model of alternating quantum walks below.

	\subsection{Alternating  quantum walks}

	Recently, an interesting model for spatial search was proposed and applied to a variety of graphs \cite{RN5,RN6,RN8}.	 In this paper, we called the model as   \textit{alternating  quantum walks}, where CTQW and the marked-vertex phase shift are alternately performed.  More specifically, two Hamiltonians are constructed: One uses the Laplacian matrix or adjacency matrix of the graph and the other uses the information of  location of the marked vertex. Then the quantum system evolves alternately under the two Hamiltonians, which is similar to the quantum approximate optimization algorithm (QAOA) \cite{farhi2014quantum}. 
	The state evolution of the system  is as follows:
	\begin{equation}
		|\vec{t},\vec{\theta}\rangle=\prod_{k=1}^{p}e^{-iHt_k}e^{-i\theta_k |m\rangle\langle m|} |s\rangle,
	\end{equation}
	where $|s\rangle$ is the uniform superposition of all vertices on a graph, $m$ is the marked vertex, and  $H$ is the Laplacian matrix or adjacency matrix of the graph. Our objective is to choose the parameters ${\vec{t},\vec{\theta}}$ so that the success probability $|\langle m|\vec{t},\vec{\theta}\rangle|^2$ is as close to 1 as possible for as small a $p$  as possible. Note that the generalized oracle $e^{-i\theta_k |m\rangle\langle m|}$ has the following effect:
	\begin{equation}
		e^{-i\theta_k |m\rangle\langle m|}\ket{v}=\left\{
		\begin{array}{rcl}
			e^{-i\theta_k}\ket{v}, & & v = m\\
			\ket{v}, & & v \not= m.
		\end{array} \right.
	\end{equation}
	It can be constructed by using the standard oracle two times as shown in \cite{RN7}. 
	The time  of this model  is defined to be the total evolution time of the two Hamiltonians:
	\begin{equation}
		T=\sum_{i=1}^{p}(t_i+\theta_i).	 
	\end{equation}

	Notably, deterministic quantum  search algorithms via alternating quantum walks were designed for the class of complete identity interdependent networks (CIINs)~\cite{RN5} and star graphs~\cite{RN8}. However, their methods are instance-specific and    difficult to generalize to other graphs. Deterministic spatial search algorithms not only mean that we can improve the theoretical success probability to $100\%$ but also imply a kind of perfect state transfer between two vertices on graphs~\cite{PhysRevA.90.012331}. Hence, this drives Ref.~\cite{RN6} to propose  the open problem ``another compelling direction for future research is making the algorithm deterministic''. Again,  how to fully characterize the class of
	graphs that permit deterministic search was proposed as a topic of future study in~\cite{RN8}.

	
	\subsection{Graph}\label{sec-graph}
	
	In this article, we only consider simple undirected connected graphs, where ``simple'' means the graph has no loops and has no  multiple edges between any two vertices. 	
	Let $G=(V, E)$ be a graph where $V$ is the vertex set and  $E$ is the edge set.  
	The Laplacian matrix of $G$  is $L=D-A$, where $D$ is the diagonal matrix with $D_{jj}= \mathrm{deg}(j)$, the degree of vertex $j$, and $A$ is the adjacency matrix of $G$.
	A graph $G=(V,E)$ is said to be {\it vertex-transitive}, if  for any two vertices $v_1$ and $v_2$ of $G$, there is an	 automorphism mapping $f:V\rightarrow V$ such that $f(v_1)=v_2$. Next we introduce some  graphs that will be discussed in this work.
	\begin{itemize}
		\item [-] Let $S$ be a set of $n$ elements and $k$ a positive integer. The vertices of the {\it Johnson graph }$J(n, k)$~\cite{regular} are  given by the $k$-subsets of $S$, with two vertices connected if their intersection has size $k-1$. 
		\item [-]Let $S$ be a set of $n$ elements and $k$ a positive integer with $n\geq 2k$. The vertices of the {\it Kneser graph } $K(n,k)$~\cite{regular} are  given by the $k$-subsets of $S$, with two vertices connected if they are disjoint.
		\item [-]Let $Q$ be a set of size $q$ and $d$ a positive integer.
		The  vertices of {\it Hamming graph } $H(d,q)$~\cite{regular} are given by the ordered $d$-tuples of elements of $Q$. Two vertices in $H(d,q)$ are adjacent if they disagree only  in one coordinate. 
		\item [-] Let $V$ be a vector space of dimension $n$ over the field $F_q$. The vertices of {\it Grassmann graph} $G_q(n, k)$~\cite{regular} are given by  the set of
		$k$-subspaces of $V$ where two $k$-subspaces $A$ and $B$ are connected if $dim(A\cap B)=k-1$.
		\item [-]The {\it rook graph} $R(m,n)$~\cite{RN6} is the graph Cartesian product $K_m \mathbin{\square} K_n$ of complete graphs $K_m$ and $K_n$, having a total of $N = mn$ vertices, where two vertices $(u, v)$ and $ (u', v')$ are adjacent  if either $u=u'$ and $(v, v')$ is an edge of $K_n$ or $v=v'$ and $(u, u')$ is an edge of $K_m$.
		\item [-]The  {\it complete-square graph}~\cite{RN6} is the graph Cartesian product of a complete graph $K_n$ and a square graph  where the square graph is a square with four vertices.  
		\item [-] The {\it complete bipartite graph} $K(N_1, N_2)$ is  a graph  whose vertex set can be partitioned into two subsets $V_1$ of size $N_1$ and $V_2$ of size $N_2$, such that there is an edge from every vertex in $V_1$ to every vertex in $V_2$ and there are no other edges. 
		
	\end{itemize}
	
	The eigenvalues of the  Laplacian matrix for the above graphs are given in  Appendix \ref{app-1} and they are all integers.
	All the above graphs except the complete bipartite graph are   vertex-transitive.   A complete bipartite graph is not  vertex-transitive  unless $N_1$=$N_2$.

	\subsection{Our contribution}
	In this article, we shall present a novel and succinct  algorithmic framework via alternating quantum walks that unifies  quantum spatial search,  perfect state transfer, and exact uniform sampling on a large class of graphs including all the graphs defined above. The framework has a succinct formalism that depends only on the  depth  of the Laplacian eigenvalue set of the graph, and is easy to use.	 We first present the algorithmic framework in Theorem \ref{main-th1}, and then as direct applications, we achieve exact uniform sampling in Corollary \ref{cor-sampling}  and perfect state transfer in Corollary \ref{cor-transfer}.    The most important and interesting application  is to tackle the spatial search problem in Theorem \ref{newth}.  When applying Corollary \ref{cor-sampling}, Corollary \ref{cor-transfer} and Theorem \ref{newth} to these graphs introduced in  Section \ref{sec-graph}, we obtain Theorem \ref{main-th2}.
	\begin{theorem}\label{main-th1} 
		Let  $G=(V,E)$ be a graph with $N$ vertices whose Laplacian matrix $L$ has only integer eigenvalues and $|s\rangle=\frac{1}{\sqrt{N}}\sum_{v\in V}  |v \rangle $. Given any $m\in V$, there is an integer $p \in O(2^{{d_L}}\sqrt{N})$ and real numbers $\gamma, \theta_j,t_j\in [0,2\pi)$   $(j\in \{1,2,\dots,p\})$, such that \begin{align} \label{main-eq}
			|s\rangle=e^{-i\gamma}\prod_{j=1}^{p}e^{-i\theta_j |m\rangle\langle m|} e^{-iLt_j}|m\rangle,
		\end{align}
		where ${d_L}$ is the depth  of the eigenvalue set of  $L$ and will be illustrated in Definition~\ref{de1}.           
	\end{theorem}
	
	The most direct application of Theorem \ref{main-th1} is exact uniform sampling. 
	\begin{corollary}
		\label{cor-sampling}
		Let  $G=(V,E)$ be a graph with $N$ vertices whose Laplacian matrix $L$ has only integer eigenvalues. There is a  quantum walk algorithm with time $O(2^{{d_L}}\sqrt{N})$  to generate exactly the uniform superposition state  over all vertices on the graph.
	\end{corollary}
	\begin{proof}
		Following Equation \eqref{main-eq}, we can generate the uniform superposition state  from any given vertex $m$   by  performing $e^{-iLt}$ and $e^{-i\theta |m\rangle\langle m|}$ alternately.   
	\end{proof}

	Moreover, by applying Theorem \ref{main-th1}, we can achieve perfect state transfer between any two
	vertices.  
	\begin{corollary}\label{cor-transfer}
		Let  $G=(V,E)$ be a graph with $N$ vertices whose Laplacian matrix $L$ has only integer eigenvalues.  Given any two vertices $u$ and $v$ on $G$, there is a quantum walk algorithm $A_{u,v}$ with time $O(2^{{d_L}}\sqrt{N})$ such that  $|\langle v |A_{u,v}|u\rangle|=1$,  which consists of three types of unitary operators: $e^{-iLt}$, $e^{-i\alpha |u\rangle\langle u|}$ and $e^{-i\beta |v\rangle\langle v|}$,  where $t$, $\alpha$ and $\beta$ are real parameters in $[0,2\pi)$.	
	\end{corollary}
	\begin{proof}
		By Theorem~\ref{main-th1}, we can use $e^{-iLt}$ and $e^{-i\alpha |u\rangle\langle u|}$($e^{-i\beta |v\rangle\langle v|}$) to construct a quantum algorithm $A_u$($A_v$) such that $A_u|u\rangle=|s\rangle$ ($A_v|v\rangle=|s\rangle$) where $|s\rangle$ is the uniform superposition state. Let $A_{u,v}=A_v^{\dagger}A_u$. Then we have
		\begin{equation}
			A_{u,v}|u\rangle=A_v^{\dagger}A_u|u\rangle=A_v^{\dagger}|s\rangle=|v\rangle.
		\end{equation}
	\end{proof}
	
	In Theorem~\ref{main-th1}, the parameters are related to  $m$, which requires  to know the position of $m$ in 	advance. However,  in the spatial search problem $m$ as a marked vertex is not known,  which prevents us from applying Theorem \ref{main-th1} to spatial search. Hence, here we propose a new theorem that additionally requires graphs to be vertex transitive. In this new theorem, the parameters are independent of $m$.
	
	\begin{theorem}\label{newth} 
		Let  $G=(V,E)$ be a vertex-transitive graph with $N$ vertices whose Laplacian matrix $L$ has only integer eigenvalues and $|s\rangle=\frac{1}{\sqrt{N}}\sum_{v\in V}  |v \rangle $. There is an integer $p \in O(2^{{d_L}}\sqrt{N})$ and real numbers $\gamma, \theta_j,t_j\in [0,2\pi)$   $(j\in \{1,2,\dots,p\})$, such that the following equation holds for all $m\in V$: \begin{align} \label{newtheq}
			|s\rangle=e^{-i\gamma}\prod_{j=1}^{p}e^{-i\theta_j |m\rangle\langle m|} e^{-iLt_j}|m\rangle,
		\end{align}
		where  ${d_L}$ is the depth  of the eigenvalue set of  $L$ and will be illustrated in Definition~\ref{de1}.           
	\end{theorem}
	
	Applying  the above results to  these specific graphs introduced in  Section \ref{sec-graph}, we immediately obtain  the following result. 
	\begin{theorem}\label{main-th2}
		Let $G$ be a graph with  a marked vertex among $N$ vertices. If it belongs to one of the following types, then there exist quantum walk algorithms for  deterministic spatial search, perfect state transfer, and exact uniform sampling on $G$ with  time  $O(\sqrt{N})$:
		\begin{enumerate}
			\item  Johnson graph $J(n,k)$ for any fixed $k$;
			\item Hamming graph $H(d,q)$ for any fixed $d$;
			\item Kneser graph $K(n,k)$ for any fixed $k$;
			\item Grassmann graph $G_q(n, k)$ for any fixed $k$;
			\item rook graph;
			\item complete-square graph;
			\item complete bipartite graph;
		\end{enumerate}
	\end{theorem}
	The  idea behind Theorem \ref{main-th2} is as follows.
	First note that all the Laplacian matrices of these graphs  have only integer eigenvalues. Therefore, by Corollary~\ref{cor-sampling} and Corollary~\ref{cor-transfer}, we can get quantum algorithms for exact uniform sampling and perfect state transfer with time  $O(2^{{d_L}}\sqrt{N})$. Next, note that the graphs of type 1 to 6  are all vertex-transitive. Then reversing the state evolution in Equation~\eqref{newtheq} results in a deterministic quantum  search algorithm  with time   $O(2^{{d_L}}\sqrt{N})$. 
	In Section \ref{sec-app-search}, we will show that $d_L$ of these graphs is independent of $N$. Thus, the  time of all the quantum algorithms is $O(\sqrt{N})$. 
	
	
	The complete bipartite graph is not vertex-transitive, which means we cannot use Theorem \ref{newth} directly, but we can still obtain deterministic quantum search algorithms by subtly adopting the idea of Theorem \ref{newth}. This will be proven in Theorem~\ref{th4} which could be worthy of more attention, since  this case is significantly different from other graphs and it may shed light on more graphs.

	The significance of our results, in our opinion, lies at least in the following  aspects:
	\begin{enumerate}[(1)]
		\item Our approach is universal in the sense that it  leads to quantum algorithms for exact uniform sampling, perfect state transfer and deterministic spatial search on a large class of graphs.   The previous uniform  sampling algorithms were usually not exact and their complexity is related to the accuracy $\epsilon$, whereas our uniform sampling algorithm is exact and requires fewer ancilla qubits. For the state transfer problem, our results reveal more graphs that permit perfect state transfer. For the spatial search problem, besides unifying and improving  plenty of previous results,  our approach   provides new results on more graphs.

		\item Our approach is succinct.  Our approach depends only on the  depth of the Laplacian eigenvalue set of the graph and then it is easy to use to design quantum algorithms.  However,  the usual method used in  existing work depends on the spectral decomposition of the walk operators in a  low-dimensional invariant subspace,  and is closely related to  topological properties of the graph. Thus, the analysis  is generally instance-specific, and  it is difficult to generalize from one graph to other graphs.
		
		\item We provide an approach to the derandomization of quantum spatial search algorithms. Due to the inherent randomness of quantum mechanics, most quantum algorithms have a certain probability of failure. It seems that everyone has habitually accepted  the inevitable failure probability and think that exact quantum algorithms may have a higher complexity (e. g., the complexity of exact quantum algorithms may be on the square or larger power of the complexity of bounded-error quantum algorithms \cite{buhrman1999bounds,exact_1,exact_2}). 
		From the perspective of computational theory,
		clarifying whether a problem can be effectively solved with certainty or only with a certain probability of success is an important theoretical issue. Our work may stimulate more discussion on this question.

	\end{enumerate}	
	
	\subsection{Technical overview }
	Although our algorithmic framework is quite succinct and does not rely on a complicated analysis  of evolution operators, there are still some non-trivial steps that involve technical challenges.
	
	The main task in this article is to prove Theorem \ref{main-th1}, which  mainly depends on  the  {\it depth}  of  the Laplacian eigenvalue set, where the depth is the number of times dividing all elements in the current set by their greatest common divisor and keeping only those elements that yield even results until the current set only contains 0.  More specifically, let $G=(V,E)$ be   a graph with Laplacian matrix $L$ whose eigenvalues are all integers and $m$ be a vertex in $V$.
	Let $\Lambda_0$ be the multiset of eigenvalues of $L$. Then for $k\geq 1 $, we iteratively  define $\Lambda_k=\{\lambda_i\in \Lambda_{k-1}\mid \frac{\lambda_i}{\mathrm{gcd}(\Lambda_{k-1})}\,is\,even\}$ until  $\Lambda_k$ contains only $0$, where $\mathrm{gcd}(\Lambda_{k-1})$ denotes the greatest common divisor of all nonzero elements in $\Lambda_{k-1}$. The depth, denoted by  $d_L$, is the  number $k$ such that $\Lambda_k$ contains only $0$.
	Let $|\eta_i\rangle$ be the eigenvector of $\lambda_i$. Each $\Lambda_k$ corresponds to a subspace $S_k=span\{|\eta_i\rangle\mid \lambda_i\in \Lambda_k\}$.  Note that $|m\rangle\in S_0$ since $S_0$ is the whole space.  Our main idea is to make the current state in $S_k$ evolve into a state in the lower-dimensional subspace $S_{k+1}$, repeating this step for each   $k\leq d_L-1$.
	
	Let $|w_k\rangle\,(k\in \{0,1,\dots,d_L\})$ be the normalized  projection component of  $|m\rangle$ onto $S_k$. Then, our aim is to achieve the following state evolution:  
	\begin{equation}
		|m\rangle = |w_0\rangle \rightarrow |w_1\rangle \rightarrow \dots \rightarrow |w_{d_L}\rangle=|s\rangle .
	\end{equation}
	In order to achieve $|w_k\rangle \rightarrow |w_{k+1}\rangle$ for each $k$, we  perform deterministic quantum search \cite{RN7} in $span\{|w_{k}\rangle ,|w_{k+1}\rangle \}$  as described in Lemma~\ref{l1}, where the main points  are as follows.
	
	\begin{enumerate}[(1)]
		\item 	Choosing an appropriate parameter  $t$ such that $e^{-iLt}$ acts as a reflection operator $U=I-2|{w}_{k+1}\rangle \langle {w}_{k+1}|$ in  $span\{|w_{k}\rangle ,|w_{k+1}\rangle \} $. 
		
		\item Constructing $V(\theta)=I- (1-e^{-i\theta})|w_k\rangle \langle w_k|$ in  $span\{|w_{k}\rangle ,|w_{k+1}\rangle \} $. When $k=0$, it is the oracle operator which can be directly used. When $k\geq 1$, we iteratively construct it. 
		\item Calculating $|\langle w_k|w_{k+1}\rangle|$ based on the location of $m$. 
	\end{enumerate}
	By Lemma \ref{l1},  there is  an integer $p \in O(\frac{1}{|\langle w_k|w_{k+1}\rangle|})$, and real numbers $\gamma$, $\theta_1,\dots,\theta_p\in[0, 2\pi)$
	such that 
	\begin{equation}
		|w_{k+1}\rangle=e^{-i\gamma}\prod_{j=1}^{p} V(\theta_j) U|w_{k}\rangle.\end{equation}
	Finally, we shall prove that $|w_{d_L}\rangle=|s\rangle$ and analyze the  complexity.

	The idea to prove Theorem \ref{newth} is similar. The difference is that when $m$ is unknown, calculating $|\langle w_k|w_{k+1}\rangle|$  in the third step will be difficult. However, we will prove that in vertex-transitive graphs, $|\langle w_k|w_{k+1}\rangle|$ is independent of the location of $m$. This attribute is   proven in Appendix~\ref{app}.

	\subsection{Related work }

 
 
\noindent{\bf Spatial search.} There has been much work devoted to the research of  quantum spatial search, and  in Table~\ref{sstab}  we list only these closely related to our work. First, quantum spatial search on Johnson graphs  has become a basis for the solution of many problems. For example,  quantum algorithms for the element distinctness problem \cite{Ambainis07}, matrix product verification \cite{BuhrmanS06}, triangle finding \cite{MagniezSS07}, and the string problem \cite{akmal2022near} are all based on  quantum walk search on Johnson graphs. 
Quantum
spatial search Johnson graphs $J(n,k)$ for any fixed $k$ was investigated in ~\cite{RN4,johnson_discrete,HG} based on CTQWs or DTQWs. A quantum search algorithm based on alternating quantum walks was proposed  for $J(n,2)$ \cite{RN6}, but it is not clear  how to extend it to the general case.
	Quantum
spatial search on hypercube graphs was  investigated in~\cite{RN2,hypercubegraph} based on CTQWs or DTQWs. Note that a hypercube graph is a special case of Hamming graph $H(d,q)$ where $q=2$. The Hamming graph $H(d,q)$ with constant $q$ and Grassmann graph $G_q(n,k)$ with constants $k$ and $q$  were investigated in~\cite{HG} based on CTQWs. The rook graph and complete square graph were  considered in~\cite{RN6} via alternating quantum walks. 	Quantum
spatial search on complete bipartite graphs was  investigated in~\cite{RN10,bipartite,PhysRevA.106.052207}  based on CTQWs or DTQWs.  Deterministic quantum search algorithms were proposed  for  complete identity interdependent networks (CIINs) ~\cite{RN5} and  for star graphs~\cite{RN8} based on  alternating quantum walks.   Note that a CIIN is a special case of rook graph $R(m,n)$ where $n=2$, and a star graph is a special case of complete bipartite graph where one partition has only one vertex. 

In order to obtain the above results, an instance-specific and complicated analysis is required in the literature, and the approach is difficult to generalize to other graphs. 
 Deterministic quantum algorithms are obtained only for CIINs and star graphs, but it is not clear how to extend them to other graphs.   Hence, this drives Ref.~\cite{RN6} to propose 
	the open problem ``another compelling direction for future research is making the algorithm deterministic''.
	How to fully characterize the class of
	graphs that permit deterministic search was proposed also as a topic of future study in~\cite{RN8}.

In this article, we obtain a quantum spatial search algorithm that deterministically find the marked vertex on all the above mentioned graphs.

	\begin{table}[ht]
		\centering
		\caption{Some related works on quantum spatial search. All the quantum algorithms  can find out a marked vertex  from $N$ vertices  in  $O(\sqrt{N})$ time (or discrete time steps). }
		\label{sstab}
		\begin{tabular}{ccc}
			\toprule
			Graphs &	Walk model    & Deterministic? \\
			\midrule
			Johnson graph & CTQW~\cite{RN4,HG}, DTQW~\cite{johnson_discrete}&  No\\
			hypercube graph	& CTQW~\cite{RN2}, DTQW~\cite{hypercubegraph}	& No \\
			Hamming graph	& CTQW~\cite{HG}	& No \\
			Grassmann graph	& CTQW~\cite{HG}	& No \\
			rook graph	& alternating quantum walks~\cite{RN6}	& No \\
			complete-square graph	&alternating quantum walks~\cite{RN6}& No \\
			complete bipartite graph& CTQW~\cite{RN10}, DTQW~\cite{bipartite,PhysRevA.106.052207} & No \\
			CIIN & alternating quantum walks~\cite{RN5} & Yes \\
			star graph & alternating quantum walks~\cite{RN8} & Yes \\
            all the above graphs& alternating quantum walks (this work) &Yes\\
			\bottomrule
		\end{tabular}

	\end{table}

	\noindent{\bf Uniform sampling.}
	Assuming $P$ is a reversible Markov chain on a graph $G=(V,E)$  with  stationary distribution $\pi$ and $|V|=N$, quantum  sampling on $P$ is to prepare a quantum state   $\ket{\pi}=\sum_{v\in V}\sqrt{\pi_v}\ket{v}$. When $\pi$ is uniform, i.e.  $\ket{\pi}=\frac{1}{\sqrt{N}}\sum_{v\in V}\ket{v}$, we say it's a uniform sampling. 
	
The quantum sampling algorithm with quadratic speed-up  was only found in special cases, e.g.,~\cite{2000quant.ph.10117N,AmbainisKV01,10.1145/380752.380757,sampling_1,sampling_2,sampling_3}. For sampling from a sequence of slowly evolving Markov chains, the running time of quantum algorithms in ~\cite{sampling_4,sampling_5}  are better than the classical algorithms under certain specific assumptions.	
For the general case,  a quantum sampling algorithm can be designed for any reversible Markov chain by using the approach  in~\cite{interpolated}. Let $HT(P,v)$ be the time from stationary distribution  to hit vertex $v$ according to  $P$. The time required to prepare a state  $\epsilon $  close to $\ket{\pi}$ is $O\left(\sqrt{HT} \frac{1}{\epsilon} \right)$, where $HT=max_v HT(P,v)$. It can be seen that the complexity of the algorithm increases rapidly with $\frac{1}{\epsilon}$. In \cite{PhysRevA.102.022423}, the complexity   was improved  to $O \left(\sqrt {HT}log \frac{1}{\epsilon} \right)$. Later, the result was  further improved in~\cite{PhysRevA.107.022432}, reducing the number of ancilla qubits. We list the cost of these quantum algorithms when  applying to uniform sampling in Table~\ref{ustab}. Compared to previous work, our quantum algorithm  prepares $\frac{1}{\sqrt{N}}\sum_{v\in V}\ket{v}$ exactly and requires no ancilla qubit. 
	\begin{table}[ht]	
		\centering
		\caption{Some related works on uniform sampling.  $HT=\Omega(N)$~\cite{aldous-fill-2014}.}
		\label{ustab}
		\begin{tabular}{cccc}
			\toprule
			Walk model &	Error &	Complexity    & Ancilla qubit \\
			\midrule
			DTQW~\cite{interpolated} &$\epsilon$ &$\Theta \left( \sqrt{HT}\ \frac{1}{\epsilon}\right)$&  $\Theta\left(  log\sqrt{HT}+log\frac{1}{\epsilon}\right)$\\
			CTQW~\cite{PhysRevA.102.022423}	&$\epsilon$ & $\Theta\left(\sqrt{HT}log\frac{1}{\epsilon}\right)$	& $\Theta\left(  log\sqrt{HT}log\frac{1}{\epsilon}\right)$ \\
			DTQW~\cite{PhysRevA.107.022432}	& $\epsilon$ &$\Theta\left(  \sqrt{HT}log\frac{1}{\epsilon}\right)$	&  $\Theta\left(  log\sqrt{HT}+log log\left(\frac{1}{\epsilon}\right)\right)$ \\
			CTQW(this work)&$0$ 	& $O\left( 2^{d_L}\ \sqrt{N}\ \right)$	& 0\\
			\bottomrule
		\end{tabular}    
	\end{table}
	
	\noindent{\bf Perfect state transfer.}
	Perfect state transfer (PST) from $u$ to  $v$ on a graph $G$  is to  construct a quantum  walk operator $U$ such that $|\langle v|U|u\rangle|=1$. A crucial problem  is to find  which graph  permits PST. There have been many research results on this problem. The conclusion of PST is not consistent in different quantum walk models.
	
	We first list some related work based on CTQWs. PST  was originally introduced by Bose~\cite{PhysRevLett.91.207901} as a tool for quantum communication through quantum spin chains.
	Then, PST in quantum spin networks was studied in~\cite{pst1,PhysRevA.71.032312}, which provided a class of qubit networks that have PST for any two-dimensional quantum state. Cheung and Godsil~\cite{CHEUNG20112468} obtained  a characterization for  cubelike graphs permitting PST. Bašić~\cite{circulant_networks} considered a  circulant graph and established a criterion for checking whether or not it permits PST. Moreover,  PST was  investigated on distance-regular graphs~\cite{COUTINHO2015108},  Hadamard diagonalizable graphs~\cite{JOHNSTON2017375} and Johnson schemes~\cite{johnsonpst}. Especially, Cayley graphs have attracted much attention since they possess distinguished algebraic structure and are closely related to group theory. PST on different  classes of Cayley graphs was investigated in~\cite{HIRANMOYPAL2016319,caleyst1,caleyst2,caleyst3,caleyst4,doi:10.1080/03081087.2022.2158163}. 
	
	In the DTQW model,  PST is mainly divided into two cases: the positions of $u$ and $v$ are known or unknown. In the first case, the dynamics of implementing PST can be  designed   globally, and this approach was analyzed on	different graphs such as a line~\cite{Yal_2015,PhysRevA.90.012331}, a circle~\cite{Yal_2015}, a 2D lattice~\cite{PhysRevA.90.012331}, a regular graph~\cite{Shang_2018}, a complete graph~\cite{Shang_2018} and a complete bipartite graph~\cite{Huang_2024}.  In the second case,  PST is usually achieved by taking the approach in quantum search, 
 and  PST was found on a star graph and a complete graph with loops~\cite{PhysRevA.94.022301} and a circulant graph~\cite{zhan_infinite_2019}. 
 
 In this paper, we consider PST in the case that $u$ and $v$ are known in advance and  alternating quantum walks are used as the walk model. Our conclusion is that PST can be achieved on a graph $G$ provided that eigenvalues of the Laplacian matrix of  $G$ are all integers.

		\subsection{Paper organization}
		The article is organized as follows. Some preliminaries are  introduced in Section \ref{se2}. Then, we prove our main theorem and algorithmic framework in Section  \ref{sec-framework}. Next, we show the applications of our algorithmic framework on some graphs  in Section \ref{sec-app-search}.  Finally, we conclude the article in Section \ref{sec6}.

		\section{Preliminaries}\label{se2}

		Below we give some properties of the  Laplacian matrix $L$ of a graph.
		\begin{lemma}\label{lemma2}
			Let $L$ be the  Laplacian matrix  of a graph $G=(V, E)$, where $|V|=N$. Then, we have the following properties.
			(\romannumeral 1) $0$ is a simple\footnote{An eigenvalue is said to be simple if its   algebraic multiplicity  is 1.} eigenvalue of $L$, and   the  corresponding eigenvector is $\frac{1}{\sqrt{N}}\sum_{v\in V}  |v \rangle$.
			(\romannumeral 2) Given the  spectral decomposition $L=\sum_{i=1}^{N} \lambda_i |\eta_i \rangle \langle \eta_i|$, $\langle v|\eta_i \rangle$ is a real number for each  $v\in V$ and each $i\in \{1,2,\dots,N\}$. 
		\end{lemma}
		\begin{proof}
			The property (\romannumeral 1) is a direct conclusion in \cite{spectral}.
			Since $D$ and $A$ are both real and symmetric, $L$ is real and symmetric and has $N$  orthogonal real eigenvectors. Thus, property (\romannumeral 2) holds.
			
		\end{proof}

		The continuous-time quantum walk (CTQW) on $G=(V, E)$ takes place
		in  the Hilbert space ${\cal H}=span\{|v\rangle: v\in V\}$ and the dynamic of the system is determined by the following Schr\"{o}dinger equation:
		\begin{equation}\label{e1}
			i\cdot \frac{d\langle v|\psi(t)\rangle}{dt}=\sum_{u\in V}\langle v|H|u\rangle \\ \langle u|\psi(t)\rangle,
		\end{equation}
		where $|\psi(t)\rangle$ denotes the state at time $t$, $H$ is a Hamiltonian satisfying $\langle v|H|u\rangle=0$ when $v$ and $u$ are not adjacent in $G$. This equation means that at time $t$, the change of the amplitude of each vertex $|v\rangle$ is only related to the amplitudes of its adjacent vertices. From Equation~\eqref{e1}, one can see that the continuous-time quantum walk over a graph $G$ at time $t$ can  be defined by the unitary transformation $U=e^{-iHt}$. In general, we let $H=\gamma L$ or $\gamma A$, where $\gamma$ is  a real and positive parameter called jumping rate. In this paper, we choose $H=L$ such that the CTQW operator is $e^{-iLt}$.
		

		To evolve $\ket{\psi_1}$ into $\ket{\psi_2}$, one frequently-used way is to perform alternately the 
		two unitary operators 
		\begin{equation}	
			\begin{aligned}
				U_1(\alpha) & =I-(1-e^{-i\alpha})|\psi_2\rangle\langle \psi_2|,\\
				U_2(\beta) & =I-(1-e^{-i\beta})|\psi_1\rangle\langle \psi_1|
			\end{aligned}
		\end{equation}
		on the initial state $|\psi_1\rangle$  where $\alpha$ and $\beta$ are real numbers. 
		Note that when $\alpha=\beta=\pi$, the above two operations recover to the reflection operations in Grover's  algorithm \cite{Grover}. In~\cite{RN9}, the authors showed that if $\alpha$ is fixed to $\pi$, given $|\langle \psi_1|\psi_2\rangle|\in (0,\frac{1}{2}]$, there exist appropriate values for $\beta$  to carry out the search deterministically and these values can be numerically calculated. In~\cite{RN7} (Theorem 1 and 2), deterministic quantum search with adjustable parameters was obtained, from which   the above result holds for any $|\langle \psi_1|\psi_2\rangle|\in (0,1)$ as follows.

		\begin{lemma}[\cite{RN7}]\label{l1}
			Given two unitary operators $U_1(\pi)$, $U_2(\beta)$, and  $|\langle \psi_1|\psi_2\rangle|\in (0,1)$, where $U_1(\pi) =I-2|\psi_2\rangle\langle \psi_2|$, $U_2(\beta)=I-(1-e^{-i\beta})|\psi_1\rangle\langle \psi_1|$, there is  an integer $p \in O(\frac{1}{|\langle \psi_1|\psi_2\rangle|})$ and real numbers $\gamma$, $\beta_1,\dots,\beta_p\in[0, 2\pi)$
			such that 
			\begin{equation}	    
				|\psi_2\rangle=e^{-i\gamma}\prod_{j=1}^{p} U_2(\beta_j)U_1(\pi)|\psi_1\rangle.
			\end{equation}
		\end{lemma}
		
		\begin{remark}
			The parameters $\beta_1,\dots,\beta_p$ in Lemma~\ref{l1} alternates between two parameters $\bar{\beta_1},\bar{\beta_2}$, i.e. $\beta_1=\bar{\beta_1}, \beta_2=\bar{\beta_2}, \beta_3=\bar{\beta_1}, \beta_4=\bar{\beta_2},\ldots$.
			The two parameters $\bar{\beta_1},\bar{\beta_2}$ can be obtained by solving a system of two trigonometric equations numerically, as the equations are too complicated to obtain a closed-form solution.
   
		\end{remark}
		
		\section{Algorithmic framework} \label{sec-framework}
		In this section,  we   present a universal algorithmic framework  unifying quantum spatial search,  perfect state transfer, and exact uniform sampling  on a large class of graphs.

		Let $S$ be a finite set of integers.  We use $\mathrm{gcd}(S)$ to denote the greatest common divisor of all nonzero elements in $S$. If there is no nonzero element in $S$, then  let $\mathrm{gcd}(S)=1$. A crucial definition is as follows.

		\begin{definition}
			Let $M$ be an $N\times N$ Hermitian matrix  with spectral decomposition $M=\sum_{i=1}^{N} \lambda_i |\eta_i \rangle \langle \eta_i|$, where $\lambda_1,\dots,\lambda_N$ are integers, at least one of which  is $0$. 
			
			(\romannumeral 1) Define $\Lambda_0=\{\lambda_1,\dots,\lambda_N\}$. For $k\geq 1$, we recursively define $\Lambda_{k}$ and $\overline{\Lambda}_{k}$ until $\Lambda_k$ contains only 0 as follows:
			\begin{equation}
				\Lambda_{k}=\{\lambda\in\Lambda_ {k-1} \mid e^{-i\lambda \frac{\pi}{\mathrm{gcd}(\Lambda_{k-1} )}}=1\},	 
			\end{equation}
			and \begin{equation}
				\overline{\Lambda}_{k}=\{ \lambda\in\Lambda_ {k-1} \mid e^{-i\lambda \frac{\pi}{\mathrm{gcd}(\Lambda_{k-1} )}}=-1\}. 
			\end{equation}
			We use $d_M$ to denote the number  $k$ such that $\Lambda_k$ contains only 0, and call it the depth of the eigenvalue set of $M$.  
			
			(\romannumeral 2) Let $|m\rangle=\sum_{i=1}^{N}\alpha_i |\eta_i\rangle$ be a vector where each $\alpha_i$ is a real number such that
			\begin{align}\label{eq-nozeoro}
				(\sum_{\lambda_i \in {\Lambda}_k} \alpha_i^2)(\sum_{\lambda_i \in \overline{\Lambda}_k} \alpha_i^2)\neq 0
			\end{align}	
			for any $k\in\{1,\dots,d_M\}$.  Let $|w_0\rangle=|m\rangle$, and for $k\in\{1,\dots,d_M\}$ define
			\begin{equation}
				|w_k\rangle =\frac{1}{\sqrt{\sum_{\lambda_i \in {\Lambda}_k} \alpha_i^2}} \sum_{\lambda_i \in {\Lambda}_k} \alpha_i |\eta_i\rangle 
			\end{equation}
			and \begin{equation}
				|\overline{w}_k\rangle =\frac{1}{\sqrt{\sum_{\lambda_i \in \overline{\Lambda}_k} \alpha_i^2}} \sum_{\lambda_i \in \overline{\Lambda}_k} \alpha_i |\eta_i\rangle.	 
			\end{equation}
			\label{de1}
		\end{definition}

		\begin{example}
			Given the following $6\times 6$ Hermitian matrix
			\begin{equation}
				\begin{aligned}
					M &=  0|\eta_1 \rangle \langle \eta_1|+1|\eta_2 \rangle \langle \eta_2|+3|\eta_3 \rangle \langle \eta_3|+6|\eta_4 \rangle \langle \eta_4|\\
					&+64|\eta_5 \rangle \langle \eta_5|
					+64|\eta_6 \rangle \langle \eta_6|,
				\end{aligned} 
			\end{equation}
			the process of computing $d_M$, $\Lambda_0$, $|w_0\rangle$, and $\Lambda_k$, $\overline{\Lambda}_{k}$, $|w_k\rangle$, $|\overline{w}_{k}\rangle$ for $k\in\{1,2,\dots,d_M\}$ with $|m\rangle=\sum_{i=1}^{6}\alpha_i |\eta_i\rangle$ where $\alpha_1,\dots,\alpha_6$ are not zero is shown in Fig.~\ref{fig1}.
			\begin{figure}[H]
				\centering
				\begin{tikzpicture} [sibling distance =140pt,level distance=80pt,thick=1, scale=1, every node/.style={scale=1}]
					\node(A)[rectangle,minimum width =20pt,minimum height =20pt ,draw=blue,align=center] {$\Lambda_0=\{0,1,3,6,64,64\}$\\$|{w}_0\rangle=|m\rangle=\sum_{i=1}^{6}\alpha_i |\eta_i\rangle$}
					child {
						node (B)[rectangle,minimum width =90pt ,minimum height =20pt ,draw=blue,align=center] {$\overline{\Lambda}_1=\{1,3\}$\\$|\overline{w}_1\rangle=\frac{\alpha_2 |\eta_2\rangle+\alpha_3 |\eta_3\rangle}{\sqrt{\alpha_2^2+\alpha_3^2}}$}				
					}
					child {node(C)[rectangle,minimum width =20pt,minimum height =20pt ,draw=blue,align=center]  {$\Lambda_1=\{0,6,64,64\}$\\$|{w}_1\rangle =\frac{\alpha_1 |\eta_1\rangle+\alpha_4 |\eta_4\rangle+\alpha_5 |\eta_5\rangle+\alpha_6 |\eta_6\rangle}{\sqrt{\alpha_1^2+\alpha_4^2+\alpha_5^2+\alpha_6^2}}$}
						child {node(D)[rectangle,minimum width =90pt,minimum height =20pt ,draw=blue,align=center]  {$\overline{\Lambda}_2=\{6\}$\\$|\overline{w}_2\rangle=\frac{\alpha_4 |\eta_4\rangle}{\sqrt{\alpha_4^2}}$}					
						}
						child {node(E)[rectangle,minimum width =20pt,minimum height =20pt ,draw=blue,align=center]{$\Lambda_2=\{0,64,64\}$\\$|{w}_2\rangle =\frac{\alpha_1 |\eta_1\rangle+\alpha_5 |\eta_5\rangle+\alpha_6 |\eta_6\rangle}{\sqrt{\alpha_1^2+\alpha_5^2+\alpha_6^2}}$}
							child {node(F)[rectangle,minimum width =90pt,minimum height =20pt ,draw=blue,align=center]  {$\overline{\Lambda}_3=\{64,64\}$\\$|\overline{w}_3\rangle=\frac{\alpha_5 |\eta_5\rangle+\alpha_6 |\eta_6\rangle}{\sqrt{\alpha_5^2+\alpha_6^2}}$}
							}
							child {node(G)[rectangle,minimum width =20pt,minimum height =20pt ,draw=blue,align=center] {$\Lambda_3=\{0\}$\\$|{w}_3\rangle =\frac{\alpha_1 |\eta_1\rangle}{\sqrt{\alpha_1^2}} $}
							}
							child [missing] {}
						}
						child [missing] {}
					}
					child [missing] {}
					;
					\begin{scope}[nodes = {draw = none}]
						\path (A) -- (B) node [near start, left]  {$\frac{\lambda_i}{\mathrm{gcd}(\Lambda_{0}) }$ is odd};
						\path (A) -- (C) node [near start, right]  {$\frac{\lambda_i}{\mathrm{gcd}(\Lambda_{0})}$ is even};
						\path (C) -- (D) node [near start, left]  {$\frac{\lambda_i}{\mathrm{gcd}(\Lambda_{1})}$ is odd};
						\path (C) -- (E) node [near start, right]  {$\frac{\lambda_i}{\mathrm{gcd}(\Lambda_{1})}$ is even};
						\path (E) -- (F) node [near start, left]  {$\frac{\lambda_i}{\mathrm{gcd}(\Lambda_{2})}$ is odd};
						\path (E) -- (G) node [near start, right]  {$\frac{\lambda_i}{\mathrm{gcd}(\Lambda_{2})}$ is even};
						\begin{scope}[nodes = {below = 20pt}]
							\node            at (G) {$d_M=3$};
						\end{scope}
					\end{scope}
				\end{tikzpicture}
				\caption{The process of computing $d_M$, $\Lambda_0$, $|w_0\rangle$ and $\Lambda_k$, $\overline{\Lambda}_{k}$, $|w_k\rangle$, $|\overline{w}_{k}\rangle$ for $k\in \{1,\dots,d_M\}$. 
				}
				\label{fig1}
			\end{figure}
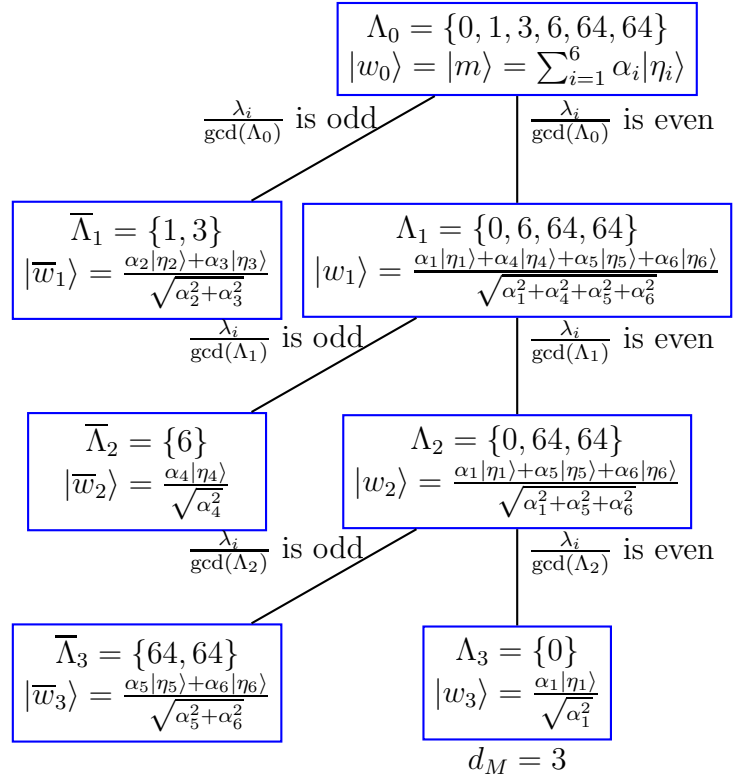
		\end{example}
		
		In the following, we consider the  class of graphs whose Laplacian matrices have only integer eigenvalues. 
		Now let $G=(V, E)$ be such a graph with $N$ vertices and  Laplacian matrix $L$. 
		Assume the  spectral decomposition of $L$ is $L=\sum_{i=1}^{N} \lambda_i |\eta_i \rangle \langle \eta_i|$.
		By Lemma~\ref{lemma2}, one  can see that $L$ satisfies the condition in Definition~\ref{de1}, and thus we can define ${d_L}$, $\Lambda_0$, and $\Lambda_k$, $\overline{\Lambda}_{k}$, $(k\in\{1,2,\dots,{d_L}\})$ as in (\romannumeral 1) of Definition~\ref{de1}.
		In the space spanned by $\{|v\rangle: v\in V\}$, $|\eta_1\rangle,\dots,|\eta_N\rangle$ constitute a set of  orthonormal  basis,
		and for  any vertex $m\in V$,  $|m\rangle$ can  be represented as $|m\rangle=\sum_{i=1}^{N}\alpha_i |\eta_i\rangle$, where each $\alpha_i=\langle \eta_i|m\rangle$ is a real number. For $|m\rangle=\sum_{i=1}^{N}\alpha_i |\eta_i\rangle$, we define $|w_0\rangle$,  $|w_k\rangle$, $|\overline{w}_{k}\rangle$ $(k\in\{1,2,\dots,{d_L}\})$ as in (\romannumeral 2) of Definition~\ref{de1}.
		
		By Lemma~\ref{lemma2}, the number $0$ is a simple eigenvalue of $L$, and $|s\rangle$ is the corresponding eigenvector. Assuming $\alpha_{*}$ is the $\alpha_i$ such that $\ket{\eta_i}=\ket{s}$. Obviously,
		\begin{equation}
			\alpha_{*}=\braket{s|m}=\frac{1}{\sqrt{N}}.
		\end{equation}
		We also have
		\begin{equation}\label{e23} 
			\begin{aligned}
				|w_{{d_L}}\rangle&=\frac{1}{\sqrt{\sum_{\lambda_i \in {\Lambda}_{d_L}} \alpha_i^2}} \sum_{\lambda_i \in {\Lambda}_{d_L}}\alpha_i |\eta_i\rangle\\
				&=\frac{\alpha_{*} }{\sqrt{ \alpha_{*}^2}}  |\eta_i\rangle=|s\rangle. 
			\end{aligned}
		\end{equation} 
  We can also see that ${d_L}$ is not bigger than the number of distinct eigenvalues of $L$, $\Lambda_{k}=\Lambda_{k+1} \cup \overline{\Lambda}_{k+1}$, $ |\overline{w}_{k+1}\rangle \in span\{|w_{k}\rangle ,|w_{k+1}\rangle \}$ and   
  $span\{|w_{k}\rangle ,|w_{k+1}\rangle \}=span\{|w_{k+1}\rangle ,|\overline{w}_{k+1}\rangle  \}$ $(k\in \{0,1,\dots,{d_L}-1\})$.  
		
		Below we present the fundamental result of this article.

		\begin{theorem} [{\normalfont Restatement of Theorem \ref{main-th1}}]
			Let  $G=(V,E)$ be a graph with $N$ vertices whose Laplacian matrix $L$ has only integer eigenvalues  and $|s\rangle=\frac{1}{\sqrt{N}}\sum_{v\in V}  |v \rangle $. Given any $m\in V$, there is an integer $p \in O(2^{{d_L}}\sqrt{N})$ and real numbers $\gamma, \theta_j,t_j\in [0,2\pi)$   $(j\in \{1,2,\dots,p\})$, such that \begin{align} 
				|s\rangle=e^{-i\gamma}\prod_{j=1}^{p}e^{-i\theta_j |m\rangle\langle m|} e^{-iLt_j}|m\rangle.		\end{align}
			\label{th1}                       
		\end{theorem}
		
		\begin{proof}
			The idea for the proof is to divide the space $span\{|v\rangle: v\in V \}$ into a series of subspaces $span\{|w_{k}\rangle ,|w_{k+1}\rangle \}$ ($k\in\{0,\dotsc,d_L-1\}$). In each subspace $ span\{|w_{k}\rangle ,|w_{k+1}\rangle \}$, we use $e^{-iLt}$ and $e^{-i\theta |m\rangle\langle m|}$  to construct $I- 2|w_{k+1}\rangle \langle w_{k+1}|$ and $I- (1-e^{-i\theta})|w_{k}\rangle \langle w_{k}|$, respectively. Then,  by Lemma \ref{l1}, there exists a procedure such that the state $|w_{k}\rangle$ evolves into $|w_{k+1}\rangle$. By repeating the above process in order, we   achieve the following state evolution:
			\begin{equation}
				|m\rangle = |w_0\rangle \rightarrow |w_1\rangle \rightarrow \dots \rightarrow |w_{d_L}\rangle=|s\rangle. 
			\end{equation}

			We first consider the continuous-time walk operator $e^{-iLt}$. Let $t=\frac{\pi}{\mathrm{gcd}(\Lambda_0)}$.  Then
			\begin{equation}\label{e10}
				\begin{aligned}
					e^{-iL  \frac{\pi}{\mathrm{gcd}(\Lambda_0)}}=\sum_{i=1}^{N} e^{-i\lambda_i  \frac{\pi}{\mathrm{gcd}(\Lambda_0)}}  |\eta_i \rangle \langle \eta_i|
					=\sum_{\lambda_i \in \Lambda_1} |\eta_i \rangle \langle \eta_i|-\sum_{\lambda_i \in \overline{\Lambda}_1} |\eta_i \rangle \langle \eta_i|.
				\end{aligned}
			\end{equation}
			Recall that  $|w_1\rangle$ ($|\overline{w}_{1}\rangle$) are linear combinations of eigenvectors $|\eta_1\rangle,\dots,|\eta_N\rangle$ whose corresponding eigenvalues are in $\Lambda_1$  ($\overline{\Lambda}_{1}$). Then it can be seen that 
			\begin{equation}
				\begin{aligned}
					&e^{-iL  \frac{\pi}{\mathrm{gcd}(\Lambda_0)}}|w_1\rangle=|w_1\rangle, \\
					&e^{-iL  \frac{\pi}{\mathrm{gcd}(\Lambda_0)}}|\overline{w}_{1}\rangle=-|\overline{w}_{1}\rangle.
				\end{aligned}
			\end{equation}
			Thus, when restricted to  subspace $span\{|w_{1}\rangle ,|\overline{w}_{1}\rangle \}=span\{|w_{0}\rangle ,|w_{1}\rangle \}$, we have
			\begin{equation}\label{e13}
				\begin{aligned}		
					e^{-iL \frac{\pi}{\mathrm{gcd}(\Lambda_0)}} & =2|{w}_{1}\rangle \langle {w}_{1}|-I=e^{i\pi}(I-2|{w}_{1}\rangle \langle {w}_{1}|) .
				\end{aligned}
			\end{equation}
			In addition, we have  
			\begin{align}\label{eqllz}
				e^{-i\theta|m\rangle \langle m| } & =I- (1-e^{-i\theta})|w_0\rangle \langle w_0| .
			\end{align}
			From Definition~\ref{de1}, for $0\leq k \leq {d_L}-1$, we have
			\begin{equation}\label{e15}
				\langle w_k|w_{k+1}\rangle=\frac{\sqrt{\sum_{\lambda_i \in \Lambda_{k+1}} \alpha_i^2}}{\sqrt{\sum_{\lambda_i \in \Lambda_k} \alpha_i^2}}.
			\end{equation}
			Since each  $\alpha_i$ can be calculated based on the position of $m$, we can get the value of $\langle w_k|w_{k+1}\rangle$. We can see that there is a $\alpha_{*}=\frac{1}{\sqrt{N}}$ in each $\Lambda_k$  since  $\Lambda_{k+1} \subseteq \Lambda_{k}$ and  $\alpha_{*}\in \Lambda_{d_L}$. Thus, we have $\langle w_k|w_{k+1}\rangle >0$. If $\langle w_k|w_{k+1}\rangle=1$, then  we do nothing.  Here we assume that $\langle w_k|w_{k+1}\rangle \in (0,1)$.
			
			According to Equations \eqref{e13}, \eqref{eqllz}, \eqref{e15} and  Lemma~\ref{l1}, there are parameters $p \in O(\frac{1}{|\langle w_{1}|w_0\rangle|})$,  $t=\frac{\pi}{\mathrm{gcd}( \Lambda_0)}$ and $\gamma$, $\theta_j\in[0, 2\pi)$ $(j\in \{1,2,\dots,p\})$ such that
			\begin{equation}
				|w_1\rangle=e^{-i\gamma}\Big(\prod_{j=1}^{p} e^{-i\theta_j |m\rangle\langle m|}e^{-iLt}\Big)|w_0\rangle.
				\label{e16}
			\end{equation}

			Next, we  prove the following fact.
			
			\begin{fact}	If there is an integer $p$ and $\gamma$, $\theta_j$,  $t_j\in[0,2\pi)$ $(j\in \{1,2,\dots,p\})$ satisfying
				\begin{equation}\label{e17}
					|w_k\rangle=e^{-i\gamma}\prod_{j=1}^{p}e^{-i\theta_j |m\rangle\langle m|} e^{-iLt_j}|w_0\rangle,
				\end{equation}
				where $1\leq k \leq {d_L}-1$, then there is an integer $p'\in O(\frac{2p}{|\langle w_k|w_{k+1}\rangle|})$ and $\gamma'$, $\theta'_j$, $t'_j\in[0,2\pi)$ $(j\in \{1,2,\dots, p'\})$ such that
				\begin{equation}	|w_{k+1}\rangle=e^{-i\gamma'}\prod_{j=1}^{p'}e^{-i\theta'_j |m\rangle\langle m|} e^{-iLt'_j}|w_0\rangle.
					\label{e18}
				\end{equation}\label{f1}
			\end{fact}
			\begin{proof} [Proof of Fact \ref{f1}]  Denote $e^{-i\gamma}\prod_{j=1}^{p}e^{-i\theta_j |m\rangle\langle m|} e^{-iLt_j}$ in Equation \eqref{e17} by $A_k$. Then we have
				\begin{equation}\label{e19}
					\begin{aligned}
						A_k e^{-i\theta |m\rangle\langle m|} A_k^{\dagger} &=  A_k(I- (1-e^{-i\theta})|w_0\rangle \langle w_0|) A_k^{\dagger} \\
						& =I- (1-e^{-i\theta})|w_k\rangle \langle w_k|.					 
					\end{aligned}
				\end{equation}
				 Restricted to  subspace  $span\{|w_{k+1}\rangle,|\overline{w}_{k+1}\rangle \}$, we have 
				\begin{equation}
					e^{-iL \frac{\pi}{\mathrm{gcd}(\Lambda_ {k})}}=2|{w}_{k+1}\rangle \langle {w}_{k+1}|-I=e^{i\pi}(I-2|{w}_{k+1}\rangle \langle {w}_{k+1}|) .
					\label{20} 
				\end{equation}
				Recall that  $span\{|w_{k}\rangle ,|w_{k+1}\rangle \}=span\{|w_{k+1}\rangle ,|\overline{w}_{k+1}\rangle  \}$ and $|\langle w_k|w_{k+1}\rangle|$ can be calculated from the position of $m$.	Therefore, by Lemma~\ref{l1} there is an integer $ p''\in O(\frac{1}{|\langle w_k|w_{k+1}\rangle|})$ and $\gamma''$, $t=\frac{\pi}{\mathrm{gcd}(\Lambda_ {k})}$, $\theta''_j\in[0,2\pi)$ $(j\in \{1,2,\dots,p''\})$ such that 
				\begin{equation}\label{e21}
					\begin{aligned}
						|w_{k+1}\rangle & =e^{-i\gamma''}\Big(\prod_{j=1}^{p''} A_k e^{-i\theta''_j |m\rangle\langle m|} A_k^{\dagger}e^{-iLt}\Big)|w_k\rangle \\
						&=e^{-i\gamma''}\Bigr(\prod_{j=1}^{p''}A_k e^{-i\theta''_j |m\rangle\langle m|} A_k^{\dagger}	e^{-iLt}\Bigl) A_k|w_0\rangle.
					\end{aligned}
				\end{equation}
				By replacing $A_k$ with $e^{-i\gamma}\prod_{j=1}^{p}e^{-i\theta_j |m\rangle\langle m|} e^{-iLt_j}$ in Equation \eqref{e21}, we can get the parameters satisfying Equation \eqref{e18}. This proves the fact.
			\end{proof}
			By Equation \eqref{e16} and using Fact \ref{f1} recursively, there is $p \in O(\prod_{k=0}^{{d_L-1}}\frac{2}{|\langle w_{k+1}|w_k\rangle|})$ and  $\gamma$, $\theta_j$, $t_j\in[0,2\pi)$ $(j\in \{1,2,\dots,p\})$ such that
			\begin{equation} 
				\begin{aligned}
					|w_{{d_L}}\rangle&=e^{-i\gamma}\prod_{j=1}^{p}e^{-i\theta_j |m\rangle\langle m|} e^{-iLt_j}|w_0\rangle\\
					&=e^{-i\gamma}\prod_{j=1}^{p}e^{-i\theta_j |m\rangle\langle m|} e^{-iLt_j}|m\rangle.
					\label{e22}
				\end{aligned}
			\end{equation}
			
			The upper bound of $p$ is
			\begin{equation}\label{e24} 
				\begin{aligned}
					\prod_{k=0}^{{d_L}-1}\frac{2}{|\langle w_{k+1}|w_k\rangle|} & =\prod_{k=0}^{{d_L}-1}\frac{2\sqrt{\sum_{\lambda_ \in \Lambda_k} \alpha_i^2}}{\sqrt{\sum_{\lambda_i \in \Lambda_{k+1}} \alpha_i^2}} \\
					& = \frac{2^{{d_L}}\sqrt{\sum_{\lambda_ \in \Lambda_0} \alpha_i^2}}{\sqrt{\sum_{\lambda_i \in \Lambda_{d_L}} \alpha_i^2}}\\
					& =2^{{d_L}}\sqrt{N},
				\end{aligned}
			\end{equation}
			where the third equation holds since $\Lambda_{d_0}$ contains all eigenvalues and $\Lambda_{d_L}$ only contains $0$ which is a simple eigenvalue with eigenvector $|s\rangle$ such that $\sum_{\lambda_i \in \Lambda_{d_L}} \alpha_i^2=\alpha_{*}^2=\frac{1}{N}$.
			Thus, the time complexity of the algorithm is:
			\begin{equation}
				T=\sum_{i=1}^{p}(t_i+\theta_i)\leq 4p\pi=O(2^{{d_L}}\sqrt{N}).	 
			\end{equation}

			This completes the proof of Theorem \ref{th1}. 
		\end{proof}

		
		The proof of Theorem~\ref{newth} is similar to the proof of Theorem~\ref{th1}. The difference is that we need to additionally prove that $\langle w_k|w_{k+1}\rangle	=\frac{\sqrt{|\Lambda_{k+1}|} }{\sqrt{|\Lambda_{k}|}}$ which is a number independent  of $|m\rangle$. More specifically, we will prove that
		\begin{equation}\label{e2}
			\begin{aligned}
				\sqrt{\sum\nolimits_{\lambda_i \in \Lambda_k} \alpha_i^2}=\sqrt{ \frac{|\Lambda_k|}{N}}, \\
				\sqrt{\sum\nolimits_{\lambda_i \in \overline{\Lambda}_k} \alpha_i^2}=\sqrt{ \frac{|\overline{\Lambda}_k|}{N}},
			\end{aligned}
		\end{equation} 
		where $|\Lambda_k|$ ($|\overline{\Lambda}_k|$) denotes the cardinality of $\Lambda_k$ ($\overline{\Lambda}_k$). The detailed proof is given in Appendix~\ref{app}. The analysis of parameters in our algorithms is detailed explained in Appendix~\ref{app_3}.

		\section{ Applications to specific graphs}\label{sec-app-search}
		In this section, 
		our aim is to show $d_L$ of  graphs of type 1 to 6 in Theorem~\ref{main-th2} is independent of $N$ and to address the case of complete bipartite graphs.  	The eigenvalues of the  Laplacian matrix for these graphs are given in  Appendix \ref{app-1} and they are all integers.
		
		
		We first take the Johnson graph as an example.

		\begin{lemma}
			\label{lm-john}
			Let $k$ be a fixed positive integer. For any Johnson graph $J(n, k)$ with  $N$ vertices, $d_L$  is independent of $N$.
		\end{lemma}		
		\begin{proof}
			The  Laplacian matrix $L$  of Johnson graph $J(n,k)$ has $min(k,n-k)+1$ distinct eigenvalues that are all integers.
			Recall that ${d_L}$ is not bigger than the number of distinct eigenvalues of $L$ and we have ${d_L}\leq min(k,n-k)+1\leq k+1$.  Since $k$ is fixed,  $d_L$  is independent of $N$. 
		\end{proof}

		The second example  is the Hamming graph.  
		\begin{lemma}		\label{th3}
			For a Hamming graph $H(d, q)$ with  $N=q^d$ vertices, $d_L$ is independent of $N$ when $d$ is fixed or $d_L=O(loglogN)$  when $q$ is fixed.
		\end{lemma}
		\begin{proof}
			The distinct eigenvalues of the Laplacian  matrix $L$ of the Hamming graph $H(d,q)$ are $0,q,2q,\dots,dq$.
			As defined in  Definition~\ref{de1}, $d_L=O(log\,d)$.  If $d$ is fixed, $d_L$ is also fixed and independent of $N$. If $q$ is fixed, $d=O(log\,N)$ and $d_L=O(logd)=O(loglogN)$.
		\end{proof}
		
		Similarly,	$d_L$ in the following graphs are all independent of $N$:  Kneser graph $K(n,k)$ for any fixed $k$, Grassmann graph $G_q(n, k)$ for any fixed $k$, rook graph  $R(m,n)$ and complete-square graph.

		Next, we consider the complete bipartite graph $K(N_1, N_2)$, which  in general is not  vertex-transitive.  	Consequently, Theorem~\ref{newth} does not apply to this situation directly. However, we can still design a deterministic quantum algorithm  by subtly adopting the idea  of Theorem~\ref{newth},  which thus shows the flexibility of our algorithmic framework. 
		
		The vertex set of a complete bipartite graph $K(N_1, N_2)$ can be partitioned into two subsets $V_1$ of size $N_1$ and $V_2$ of size $N_2$, such that there is an edge from every vertex in $V_1$ to every vertex in $V_2$ and there are no other edges. A special case of complete bipartite graphs is star graphs, where there is only one vertex in $V_1$ or $V_2$. For star graphs, Ref.~\cite{RN8} proposed  a  deterministic quantum search algorithm via alternating quantum walks.  We will give a generalized algorithm for any complete bipartite graph.
		
		\begin{theorem}\label{th4}
			Given a complete bipartite graph $K(N_1, N_2)$ with a marked vertex $m$, there is a quantum search algorithm via alternating quantum walks that deterministically finds the marked vertex using   time $O(\sqrt{N_1+N_2})$ .
		\end{theorem}
		\begin{proof}
			The adjacency matrix $A$ of a complete bipartite graph $K(N_1, N_2)$ has three distinct eigenvalues: 0, $\sqrt{N_1N_2}$, $-\sqrt{N_1N_2}$. The algebraic multiplicity of 0 is $N_1+N_2-2$, and both $\sqrt{N_1N_2}$ and $-\sqrt{N_1N_2}$ are simple eigenvalues. Moreover,  $A$ has the following spectral decomposition
			\begin{equation}\label{e25}
				\begin{aligned}
					A=\sqrt{N_1N_2}|\eta_+\rangle \langle \eta_+| -\sqrt{N_1N_2}|\eta_-\rangle \langle \eta_-|	+\sum_{i=1}^{N_1+N_2-2}0|\eta_i\rangle \langle \eta_i|, 
				\end{aligned}
			\end{equation}
			where 
			\begin{equation} \label{e26}
				\begin{aligned}
					&\ |\eta_+\rangle=\frac{1}{\sqrt{2N_1N_2}}(\underbrace  {\sqrt{N_2},\dots,\sqrt{N_2}}_{N_1},\underbrace{\sqrt{N_1},\dots,\sqrt{N_1}}_{N_2})^T,\\	&\ |\eta_-\rangle=\frac{1}{\sqrt{2N_1N_2}}(\underbrace  {\sqrt{N_2},\dots,\sqrt{N_2}}_{N_1},\underbrace{-\sqrt{N_1},\dots,-\sqrt{N_1}}_{N_2})^T.		 
				\end{aligned}
			\end{equation}
			The marked vertex  is represented in this basis as
			\begin{equation}
				|m\rangle=\sum_{i=1}^{N_1+N_2-2}\alpha_i |\eta_i\rangle+\alpha_+|\eta_+\rangle+\alpha_-|\eta_-\rangle.
				\label{e27}
			\end{equation}
			If the marked vertex is in $V_1$, then 
			\[\langle\eta_+|m\rangle=\langle\eta_-|m\rangle=\frac{1}{\sqrt{2N_1}},\]
			and
			\begin{equation}
				|m\rangle=\sum_{i=1}^{N_1+N_2-2}\alpha_i |\eta_i\rangle+\frac{1}{\sqrt{2N_1}}|\eta_+\rangle+\frac{1}{\sqrt{2N_1}}|\eta_-\rangle.
				\label{e28}
			\end{equation}
			Define 
			\begin{equation}
				\begin{aligned}
					&\	|\eta_0\rangle=\frac{1}{\sqrt{\sum_{i=1}^{N_1+N_2-2}\alpha_i^2}}\sum_{i=1}^{N_1+N_2-2}\alpha_i |\eta_i\rangle, \\
					&\	|s\rangle=\frac{1}{\sqrt{2}}(|\eta_+\rangle+|\eta_-\rangle)=\frac{1}{\sqrt{N_1}}(\underbrace{1,1,\cdots,1}_{N_1},0,\cdots,0).
					\label{e29}
				\end{aligned}
			\end{equation}
			We can see in subspace $ span\{|\eta_0\rangle,|s\rangle\}=span\{|m\rangle,|s\rangle\}$,
			\begin{equation}
				e^{-iA\frac{\pi}{\sqrt{N_1N_2}}}=I-2|s\rangle \langle s|.
				\label{e30}
			\end{equation}
			Obviously, $\langle m|s\rangle=\frac{1}{\sqrt{N_1}}$. Thus,  by Lemma~\ref{l1} there are parameters $p \in O(\sqrt{N_1})$, $\gamma$, $\theta_j\in[0, 2\pi)$ $(j\in \{1,2,\dots,p\})$ such that
			\begin{equation}
				|s\rangle=e^{-i\gamma}\prod_{j=1}^{p}e^{-i\theta_j |m\rangle\langle m|}e^{-iA\frac{\pi}{\sqrt{N_1N_2}}}|m\rangle.
				\label{e31}
			\end{equation} 
			
			Hence, by reversing the above state evolution, we obtain a quantum search algorithm that uses  time $O({\sqrt{N_1}})$ and finds the marked vertex with certainty from the uniform superposition state over all vertices in $V_1$. Similarly, if the marked vertex is in $V_2$, we can construct a quantum search with  time  $O({\sqrt{N_2}})$. Thus, by running the two algorithms  in order, we have a quantum algorithm that has time $O({\sqrt{N_1}+\sqrt{N_2}})$   and finds the marked vertex with certainty. Since ${\sqrt{N_1}+\sqrt{N_2}}\leq \sqrt{2N_1+2N_2}$, the algorithm has time $O(\sqrt{N_1+N_2})$.
		\end{proof}

		\section{Conclusion and Discussion}\label{sec6}
		In this article, we have presented an quantum algorithmic framework, (i.e., Theorem \ref{main-th1}),  that  unifies quantum spatial search,  state transfer and uniform sampling  on a large class of graphs.  Using this framework, we can achieve exact uniform sampling
		over all vertices and perfect state transfer between any two vertices on a graph,  provided that the graph's Laplacian matrix has only integer eigenvalues. Furthermore, if the graph is vertex-transitive as well, we can achieve deterministic quantum spatial search that finds a marked vertex with certainty.  Applying these results to several  kinds of graphs often studied in the existing work, such as Johnson graphs, Hamming graphs, Kneser graphs, Grassmann graphs, rook graphs, complete-square graphs and complete bipartite graphs,  we immediately obtain quantum algorithms of deterministic  spatial search, exact uniform sampling, and perfect state transfer on these graphs with time $O(\sqrt{N})$, where $N$ is the number of vertices of the graph.
		
		Compared with existing work, our work  not only has provided a succinct and universal method, but also has obtained more optimized results. The method is easy to use since it has a succinct formalism that depends only on the depth of the Laplacian eigenvalue set of the graph. Also, this method can uniformly resolve the three different problems (spatial search,  state transfer and uniform sampling), whereas existing works resolve these problem one by one.
		Our work has improved the previous results. For instance, 
		the previous uniform sampling algorithms were usually not exact and their complexity is related to the accuracy, whereas our uniform sampling algorithm is
		exact and requires fewer ancilla qubits. For the state transfer problem, our results reveal more graphs that permit perfect state transfer. For the spatial search problem,
		besides unifying and improving plenty of previous results, our approach provides new results on more graphs.

		We have provided an approach to the derandomization of quantum spatial search algorithms.  Due to the inherent randomness of quantum mechanics, most quantum algorithms have a certain probability of failure. It seems that everyone has habitually accepted the inevitable probability of failure in quantum algorithms. There is a lack of in-depth thinking about the following question: What problems can quantum algorithms in principle solve efficiently and deterministically? Our work may stimulate more discussion on this question.

		Theorem \ref{th4} dedicated to the case of complete bipartite graphs is worthy of more attention, since  the case is significantly different from  the others. Despite this different,  we can still deal with the case  by subtly adopting the idea of Theorem \ref{newth}, which thus shows the flexibility of Theorem \ref{newth}.
		For future work, we hope to extend our results to more general graphs and more problems. 
		
		
		\bibliographystyle{alpha}
		\bibliography{references.bib}
		
		
	\appendix
	
	\section{Eigenvalues associated with some typical graphs }\label{app-1}
	
	The Johnson graph $J(n, k)$ has Laplacian eigenvalues $ni+i-i^2$ with multiplicity $\binom{n}{i}-\binom{n}{i-1}$\footnote{ $\binom{n}{-1}$ is defined as $0$ for each $n$.} where $i=0,1,2,\dots,min(k,n-k)$.
	
	The Kneser graph  $K(n,k)$~\cite{regular} has Laplacian eigenvalues $\binom{n-k}{k}-(-1)^i\binom{n-k-i}{k-i}$ with multiplicity $\binom{n}{i}-\binom{n}{i-1}$ where $i=0,1,2,\dots,k$.
	
	The Hamming graph $H(d,q)$~\cite{regular} has Laplacian eigenvalues $qi$ with multiplicity $\binom{d}{i}(q-i)^i$ where $i=0,1,2,\dots,d$.
	
	The  Grassmann graph $G_q(n, k)$~\cite{regular} has Laplacian eigenvalues $q[k]_q[n-k]_q-q^{i+1}[k-i]_q[n-k-i]_q+[i]_q$ where $i=0,1,2,\dots,min(k,n-k)$ and $[k]_q=\frac{q^k-1}{q-1}$.
	
	The rook graph $R(m,n)$~\cite{RN6} has  Laplacian  eigenvalues $0,m,n,m+n$.

	The $K_n \mathbin{\square} Q_2$ graph~\cite{RN6} that is called complete-square graph has  Laplacian  eigenvalues $0,2,4,n,n + 2,n + 4$.
	
	The   complete bipartite graph $K(N_1, N_2)$ has Laplacian eigenvalues $0$, $n$, $m$, $n+m$. The  adjacency matrix $A$ of $K(N_1, N_2)$   has three distinct eigenvalues: 0, $\sqrt{N_1N_2}$, $-\sqrt{N_1N_2}$. The algebraic multiplicity of 0 is $N_1+N_2-2$, and both $\sqrt{N_1N_2}$ and $-\sqrt{N_1N_2}$ are simple eigenvalues.

	

	\section{Proof of Equation~(\ref{e2})}\label{app}
	Since $G$ is vertex-transitive, given any vertex $v$ in $G$, there is an automorphism mapping $f:V\rightarrow V$ such that $f(m)=v$. We define $S= span\{|v_1\rangle,\dots,|v_N\rangle\}$, and a linear mapping $g:S\rightarrow S$ such that $g(|u_1\rangle)=|f(u_1)\rangle,\,g(|u_1\rangle+|u_2\rangle)=g(|u_1\rangle)+g(|u_2\rangle)$, and $g(\beta |u_1\rangle)=\beta g(|u_1\rangle)$ where $u_1,u_2$ are any two vertices in $G$ and $\beta$ is any real number. 
	Recall that $|m\rangle=\sum_{i=1}^{N} \alpha_i |\eta_i\rangle$.  We have
	\begin{equation}\label{e35}
		|v\rangle=|f(m)\rangle=g(|m\rangle)=\sum_{i=1}^{N} \alpha_i\,g(|\eta_i\rangle)
	\end{equation}
	Next we shall prove that if $L\,|\eta_i\rangle=\lambda_i\,|\eta_i\rangle$, then $L\,g(|\eta_i\rangle)=\lambda_i\,g(|\eta_i\rangle)$. Recall that $f$ is an automorphism mapping and we have
	
	\begin{equation} 
		\langle u_1|\,L\,|u_2\rangle=\langle f(u_1)|\,L\,|f(u_2)\rangle=g(|u_1\rangle)^T L \, g(|u_2\rangle)
	\end{equation}  where $u_1,u_2$ are any two vertices in $G$. Assuming that $|\eta_i\rangle=\sum_{k=1}^{N} \beta_{k} |v_k\rangle$ where $v_k$ is a vertex in $G$ and the following equation holds for each $j\in\{1,2,\dots,N\}$.
	
	\begin{equation}\label{e33}
		\begin{aligned}
			&g(|v_j\rangle)^T\,L\,g(|\eta_i\rangle)=g(|v_j\rangle)^T\,L\,g(\sum_{k=1}^{N} \beta_k |v_k\rangle)=\sum_{k=1}^{N} \beta_k\, g(|v_j\rangle)^TL\, g(|v_k\rangle)\\
			&=\sum_{k=1}^{N} \beta_k \langle v_j|\,L\, |v_k\rangle= \langle v_j|\,L\, \sum_{k=1}^{N} \beta_k|v_k\rangle=\langle v_j|\,L\, |\eta_i\rangle=\langle v_j|\,\lambda_i\,|\eta_i\rangle=\lambda_i \beta_j.
		\end{aligned}
	\end{equation}
	From (\ref{e33}), we can obtain that
	\begin{equation}\label{e34}
		L\,g(|\eta_i\rangle)=\sum_{j=1}^{N}\lambda_i\, \beta_j\, g(|v_j\rangle)=\lambda_i\,g(\sum_{j=1}^{N}\beta_j\,|v_j\rangle)=\lambda_i\,g(|\eta_i\rangle)
	\end{equation}
	This equation shows that $g(|\eta_i\rangle)$ and $|\eta_i\rangle$ are the eigenvectors of  the same eigenvalue and $g(|\eta_i\rangle)$ are also orthonormal   since $g$ is linear and $|\eta_i\rangle$ are orthonormal.
	Let $\lambda'_1,\lambda'_2,\dots,\lambda'_n$ be the distinct values of  $\lambda_1,\lambda_2,\dots,\lambda_N$.
	From (\ref{e35}) and (\ref{e34}),  and the modulus of the projection vectors of $|v\rangle$ and $|m\rangle$ on the eigenspace of $\lambda'_j$  are equal   for each $j\in\{1,2,\dots,n\}$ and $v$ in $G$. This fact implies that $\sqrt{\sum\nolimits_{\lambda_i= \lambda'_j} \alpha_i^2}$ is independent from the location of $m$ for each $j$.
	The Definition~\ref{de1} ensures that the equal eigenvalues are in a same set and we have $\sqrt{\sum\nolimits_{\lambda_i \in \Lambda_k} \alpha_i^2}$ and $\sqrt{\sum\nolimits_{\lambda_i \in \overline{\Lambda}_k}\alpha_i^2}$ are both independent from the location of m. Next we give the specific values of them. Assuming the algebraic multiplicity of $\lambda'_j$ is $a_j$ and its orthonormal eigenvectors are $|\eta'_1\rangle,|\eta'_2\rangle,\dots,|\eta'_{a_j}\rangle$. We have
	\begin{equation}\label{e36}
		\begin{aligned}
			a_j&=\sum_{k=1}^{a_j} \langle \eta'_k|\eta'_k\rangle=\sum_{k=1}^{a_j} \langle \eta'_k|\,\sum_{v\in V}|v\rangle\langle v|\,|\eta'_k\rangle=\sum_{v\in V} \sum_{k=1}^{a_j} \langle \eta'_k|v\rangle\langle v|\eta'_k\rangle\\
			&=N\sum_{k=1}^{a_j} \langle \eta'_k|m\rangle\langle m|\eta'_k\rangle=
			N\sum\nolimits_{\lambda_i= \lambda'_j} \alpha_i^2,
		\end{aligned}
	\end{equation}
	where the  first equation in the second line holds since $\sum_{k=1}^{a_j} \langle \eta'_k|v\rangle\langle v|\eta'_k\rangle$ is the square of module of the projection vectors of $|v\rangle$  on the eigenspace of $\lambda'_j$ and this is a fixed number for each $v$.
	We can see that
	\begin{equation}
		\sqrt{\sum\nolimits_{\lambda_i \in \Lambda_k} \alpha_i^2}=\sqrt{\sum\nolimits_{\lambda'_j\in \Lambda_k}\sum\nolimits_{\lambda_i=\lambda'_j} \alpha_i^2}=\sqrt{\sum\nolimits_{\lambda'_j\in \Lambda_k}\frac{a_j}{N}}=\sqrt{ \frac{|\Lambda_k|}{N}}.
	\end{equation}
	
	Similarly,
	\begin{equation}
		\sqrt{\sum\nolimits_{\lambda_i \in \overline{\Lambda}_k} \alpha_i^2}=\sqrt{ \frac{|\overline{\Lambda}_k|}{N}}.
	\end{equation}

\section{Analysis of parameters in our algorithms}\label{app_3}
In Theorem~\ref{th1}, we achieve $\ket{m}=\ket{w_0}\rightarrow \ket{w_{d_L}}=\ket{s}$ by achieving $\ket{w_k}\rightarrow \ket{w_{k+1}}$ for each $k$ using Lemma~\ref{l1}. More specifically, in each $span\{|w_{k}\rangle ,|w_{k+1}\rangle \}$, we constructed two types of unitary operators:
\begin{equation}
	\begin{aligned}
		&U_1(\pi) =I-2|w_{k+1}\rangle\langle w_{k+1}|\\
		&U_2(\beta)=I-\left(1-e^{-i\beta}\right)|w_k\rangle\langle w_k|.
	\end{aligned}
\end{equation}
Using Theorem~\ref{l1},  there are parameters $p,\gamma,\beta_j\left(j\in \{1,\dots,p\}\right)$ such that
\begin{equation}\label{eqapp_1}
	|w_{k+1}\rangle=e^{-i\gamma}\prod_{j=1}^{p} U_2\left(\beta_j\right)U_1\left(\pi\right)|w_k\rangle.
\end{equation}
The global phase $\gamma$ can be ignored and what we  care about are $p,\beta_j\left(j\in \{1,\dots,p\}\right)$. The reference~\cite{RN7} of Lemma~\ref{l1} list the conditions that these parameters need to satisfy. Let
\begin{equation}
	\sqrt{\lambda}=\left|\braket{w_k|w_{k+1}}\right|.
\end{equation}
The condition that the integer $p$ needs to satisfy is
\begin{flalign}
	\left\{
	\begin{array}{rcl}
		&p\text{ is even},\\
		&p> \frac{\pi}{|8arcsin\sqrt{\lambda}|}.
	\end{array} \right.  
\end{flalign}
For each $p$ satisfies the above condition, there are $\beta_j$ $(j\in \{1,\dots,p\})$ satisfying  Equation~(\ref{eqapp_1}), where
\begin{flalign}
	\beta_j=	\left\{
	\begin{array}{rcl}
		\beta'_1 & & j\text{ is odd}\\
		\beta'_2 & & j\text{ is even}\\
	\end{array} \right.^{} ,
\end{flalign}
and $\beta'_1,\beta'_2 $ are solutions to the following  equations:
\begin{flalign}
	\left\{
	\begin{array}{rcl}
		&	cos\frac{\psi}{2}=-cos(\frac{\beta'_1+\beta'_2}{2})-8\lambda sin\frac{\beta'_1}{2}sin\frac{\beta'_2}{2}+8\lambda^2 sin\frac{\beta'_1}{2}sin\frac{\beta'_2}{2},\\ 
		&	0=sin\frac{\psi}{2}cos(p\psi)-4\lambda(1-2\lambda) sin(p\psi) sin\frac{\beta'_1}{2}sin\frac{\beta'_2}{2},\\ 
		& 0=(4\lambda-1)sin\frac{\beta'_1}{2}cos\frac{\beta'_2}{2}-cos\frac{\beta'_1}{2}sin\frac{\beta'_2}{2}.\\ 
	\end{array} \right.
\end{flalign}
We can see that  the equations are too complicated to obtain a closed-form solution. To design parameters in our algorithms more simply, we propose an improvement method which doesn't increase the time complexity and only requires an ancilla qubit. Our main idea is use the following theorem to achieve $\ket{w_k}\rightarrow \ket{w_{k+1}}$.
\begin{theorem}[\cite{PhysRevA.64.022307}]\label{app_th1}
Given $U_1\left(\alpha\right)$,  $U_2\left(\beta\right)$ and $|\langle \psi_1|\psi_2 \rangle|\in (0,1)$ where $U_1\left(\alpha\right) =I-\left(1-e^{-i\alpha}\right)|\psi_2 \rangle\langle \psi_2 |$, $U_2\left(\beta\right)=I-\left(1-e^{-i\beta}\right)|\psi_1\rangle\langle \psi_1|$. There is $p \in O\left(\frac{1}{|\langle \psi_1|\psi_2\rangle|}\right)$ and $\gamma$, $\alpha_j,\beta_j\left(j\in \{1,\dots,p\}\right)\in[0, 2\pi)$ such that
\begin{equation}\label{eqapp_2}
	|\psi_2\rangle=e^{-i\gamma}\prod_{j=1}^{p} U_2\left(\beta_j\right)U_1\left(\alpha\right)|\psi_1\rangle.
\end{equation}
If $p$ satisfies that \begin{equation}
	p\geq \frac{\pi-2\arcsin\left(|\langle \psi_1|\psi_2\rangle|\right)}{4\arcsin\left(|\langle \psi_1|\psi_2\rangle|\right)},
\end{equation}
there are parameters $\alpha_j,\beta_j\left(j\in \{1,\dots,p\}\right)$ satisfying  Equation~(\ref{eqapp_2}), where
\begin{equation}
	\begin{aligned}
		\alpha_j=\beta_j=2arcsin\left(\frac{sin\left(\frac{\pi}{4p+6}\right)}{|\langle \psi_1|\psi_2\rangle|}\right),\quad j\in \{1,\dots,p\}.
	\end{aligned}
\end{equation}
\end{theorem}
It's obvious that if we use Theorem~\ref{app_th1} to achieve $\ket{w_k}\rightarrow \ket{w_{k+1}}$, parameters in our algorithms  will have analytical expressions.
According to Theorem~\ref{app_th1}, to achieve $\ket{w_k}\rightarrow \ket{w_{k+1}}$ in $span\{|w_{k}\rangle ,|w_{k+1}\rangle \}$, we need   these three conditions:
\begin{enumerate}
	\item Constructing $U_1\left(\alpha\right)=I-\left(1-e^{-i\alpha}\right)|w_{k+1}\rangle\langle w_{k+1}|$.
	\item Constructing $U_2\left(\beta\right)=I-\left(1-e^{-i\beta}\right)|w_k\rangle\langle w_k|$.
	\item Knowing the value of $\braket{w_k|w_{k+1}}$.
\end{enumerate}
In the proof of Theorem~\ref{th1}, we have proven that the second and third conditions can be satisfied. For the first condition, we construct a speical case of $U_1\left(\alpha\right)$ where $\alpha=\pi$:
\begin{equation}
	U_1\left(\pi\right)=I-2|w_{k+1}\rangle\langle w_{k+1}|.
\end{equation}
Next we describe that how to use $U_1\left(\pi\right)$ to construct $U_1\left(\alpha\right)$  for any $\alpha$.

\begin{theorem}\label{app_th2}
Let $\ket{\psi}$ be any quantum state,	and  $U$ satisfies that $U|\psi\rangle=-|\psi\rangle$ and $U|\psi^\perp\rangle=|\psi^\perp\rangle$, where $|\psi^\perp\rangle$ is orthogonal with $|\psi\rangle$. Given any $\theta\in[0,2\pi)$, we can construct a operator $A$, where $A$ calls  $U$ twice and satifies that $A|\psi\rangle=e^{i\theta}|\psi\rangle$ and $A|\psi^\perp\rangle=|\psi^\perp\rangle$. \end{theorem}

\begin{proof}
The quantum circuit of $A$ can be constructed as shown in Figure~\ref{fig3.2}.
	\begin{figure}[htbp]
		\centering
		\begin{quantikz} 
			\lstick{$\ket{0}$} & \gate{H} & \ctrl{1} & \gate{H} & \gate{Z_{\theta}} & \gate{H} & \ctrl{1} & \gate{H} &\qw \\
			& \qwbundle[alternate]{} & \gate{U}\qwbundle[alternate]{} & \qwbundle[alternate]{} & \qwbundle[alternate]{}& \qwbundle[alternate]{} &\gate{U}\qwbundle[alternate]{}  & \qwbundle[alternate]{} &\qwbundle[alternate]{}\\
		\end{quantikz}
		\caption{The quantum circuit
			of $A$ in Theorem~\ref{app_th2}}
		\label{fig3.2}
	\end{figure}
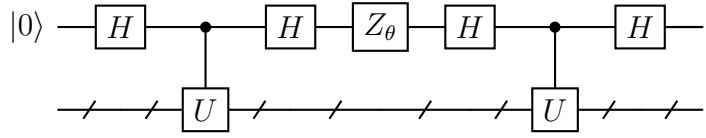\\
	In the quantum circuit, the first line represents an  ancilla qubit initialized to $\ket{0}$ and the second line  represents qubits which $U$ acts on. The gate
	\begin{equation}
		H= \begin{pmatrix}
			\sqrt{\dfrac{1}{2}} & \sqrt{\dfrac{1}{2}} \\
			\sqrt{\dfrac{1}{2}} & -\sqrt{\dfrac{1}{2}}
		\end{pmatrix},
	\end{equation}
    and
	\begin{equation}
		Z_{\theta}= \begin{pmatrix}
			1 & 0 \\
			0 & e^{i\theta}
		\end{pmatrix}.
	\end{equation}

After simple calculation, it can be obtained that
\begin{equation}
				A\ket{0}\ket{\psi}=e^{i\theta} \ket{0} \ket{\psi},
	\end{equation}
and 
\begin{equation}
	A\ket{0}\ket{\psi^\perp}=\ket{0} \ket{\psi^\perp}.
\end{equation}
\end{proof}
 By the above theorem, we can use
\begin{equation}
	U_1\left(\pi\right)=I-2|w_{k+1}\rangle\langle w_{k+1}|
\end{equation}
to construct $A\left(\alpha\right)$ such that 
\begin{equation}
	\begin{aligned}
		&A\left(\alpha\right)|w_{k+1}\rangle=e^{-i\alpha}|w_{k+1}\rangle.\\
		&A\left(\alpha\right)|w_{k+1}^\perp\rangle=|w_{k+1}^\perp\rangle.
	\end{aligned}
\end{equation}
In  subspace $span\{|w_{k+1}\rangle ,|\overline{w}_{k+1}\rangle \}=span\{|w_{k}\rangle ,|w_{k+1}\rangle \}$, the effect of  $A(\alpha)$ is equivalent to
\begin{equation}
		I-\left(1-e^{-i\alpha}\right)|w_{k+1}\rangle\langle w_{k+1}|=U_1\left(\alpha\right).
\end{equation}
Thus, we satisfy all three conditions in Theorem~\ref{app_th1} and we can use it to achieve $\ket{w_k}\rightarrow \ket{w_{k+1}}$:
If $p$ satisfies that \begin{equation}
	p\geq \frac{\pi-2\arcsin\left(\left|\braket{w_k|w_{k+1}}\right|\right)}{4\arcsin\left(\left|\braket{w_k|w_{k+1}}\right|\right)},
\end{equation}
there are parameters $\gamma$ and
\begin{equation}
	\begin{aligned}
		\alpha_j=\beta_j=2arcsin\left(\frac{sin\left(\frac{\pi}{4p+6}\right)}{\left|\braket{w_k|w_{k+1}}\right|}\right),\quad j\in \{1,\dots,p\},
	\end{aligned}
\end{equation}
such that
\begin{equation}
	|w_{k+1}\rangle=e^{-i\gamma}\prod_{j=1}^{p} U_2\left(\beta_j\right)U_1\left(\alpha_j\right)|w_k\rangle.
\end{equation}
In this way, we provide an analytical expression for  parameters in our algorithms.
\end{document}